\documentclass[11pt]{article}
\usepackage{graphicx,enumerate,amsmath,amsthm}

\newtheorem{thm}{Theorem}[section]
\theoremstyle{definition}

\theoremstyle{remark}

\newtheorem{obs}{Observation}[section]
\theoremstyle{plain}
\newtheorem{lem}[thm]{Lemma}

\begin{document}

\title{Faster Randomized Worst-Case Update Time for Dynamic Subgraph Connectivity}
\author{Ran Duan\\ Tsinghua University\\ \footnotesize\texttt{duanran@mail.tsinghua.edu.cn}
	\and Le Zhang\\ Tsinghua University\\
\footnotesize\texttt{le-zhang12@mails.tsinghua.edu.cn}}
\date{}

\maketitle

\begin{abstract}
	Real-world networks are prone to breakdowns. Typically in the underlying graph $G$, besides the insertion or deletion of edges, the set of active vertices changes overtime. A vertex might work actively, or it might fail, and gets isolated temporarily. The active vertices are grouped as a set $S$. $S$ is subjected to updates, i.e., a failed vertex restarts, or an active vertex fails, and gets deleted from $S$. Dynamic subgraph connectivity answers the queries on connectivity between any two active vertices in the subgraph of $G$ induced by $S$. The problem is solved by a dynamic data structure, which supports the updates and answers the connectivity queries. In the general undirected graph, the best results for it include $\widetilde{O}(m^{2/3})$ deterministic \emph{amortized} update time, $\widetilde{O}(m^{4/5})$ and $\widetilde{O}(\sqrt{mn})$ deterministic \emph{worst-case} update time. In the paper, we propose a randomized data structure, which has $\widetilde{O}(m^{3/4})$ worst-case update time.
\end{abstract} 

\section{Introduction}
Dynamic subgraph connectivity is defined as follows: Given an undirected graph $G=(V,E)$ with $|E|=\Omega(V)$, there is a subset $S\subseteq V$. $E$ is subjected to edge updates of the forms $insert(e,E)$ or $delete(e,E)$, where $e$ is an edge. Besides, given $v\in V$, there are vertex updates of the forms $insert(v,S)$ or $remove(v,S)$. Through vertex updates, $S$ changes overtime. We want to query whether any two vertices $s$ and $t$ are connected in the subgraph of $G$ induced by $S$, i.e.~a query of the form $connected(s,t,S)$ for $s,t\in S$.

The problem was first proposed by Frigioni and Italiano \cite{frigioni2000dynamically}, where it was referred to as \emph{complete dynamic model}, and poly-logarithmic algorithms on connectivity were described for the special case of planar graphs. As to the general graphs Chan \cite{chan2002dynamic} first described an algorithm of deterministic amortized update time $\widetilde{O}(m^{4\omega/(3\omega+3)})$\footnote{$\widetilde{O}(\cdot)$ hides poly-logarithmic factors.}, where $m$ is the number of edges, and $\omega$ is the matrix multiplication exponent. Adopting FMM (\emph{Fast Matrix Multiplication}) algorithm of \cite{Coppersmith1990251}, the update time is $O(m^{0.94})$. Its query time and space complexity are $\widetilde{O}(m^{1/3})$ and linear, respectively. Following it, Chan, P\v{a}tra\c{s}cu, and Roditty \cite{chan2011dynamic} proposed a simpler algorithm, having the improved update time of $\widetilde{O}(m^{2/3})$. The space complexity of the new algorithm increases to $\widetilde{O}(m^{4/3})$. The new algorithm is of compact description, getting rid of the use of FMM. With the same update time, Duan \cite{duan2010new} presented new data structures occupying linear space. The worst-case udpate time algorithm is first considered by Duan \cite{duan2010new}. The algorithm proposed has the deterministic worst-case update time of $\widetilde{O}(m^{4/5})$. Via an application of dynamic DFS tree \cite{baswana2016dynamic}, Baswana et al.~proposed a new algorithm with $\widetilde{O}(\sqrt{mn})$ deterministic worst-case update time. Its query time is $O(1)$. Recently it has an improvement \cite{DBLP:journals/corr/ChenDWZ16}. These results are summarized in Table \ref{table_results}.

A close related problem is dynamic graph connectivity, which cares only about the edge updates. Poly-logarithmic amortized update time was first achieved by Henzinger and King \cite{henzinger1999randomized}. The algorithm proposed is randomized Las Vegas. Inspired by it, Holm et al.~\cite{holm2001poly} proposed a deterministic algorithm with $O(\lg ^2 n)$\footnote{We use $\lg x$ to denote $\log_2 x$.} amortized update time, which is now one of the classic results in the field. A cell-probe lower bound of $\Omega(\lg n)$ per operation was proved by P\v{a}tra\c{s}cu and Demaine \cite{DynamicConnectivity_SICOMP}. The lower bound is amortized randomized. Near-optimal results were considered by Thorup \cite{thorup2000near}, where a randomized Las Vegas algorithm was described with $O(\lg n (\lg\lg n)^3)$ amortized update time. The upper bound is recently improved to $O(\lg n (\lg\lg n)^2)$ by Huang et al.~\cite{huang2016}. Besides the classic deterministic $O(\lg ^2 n)$ result, a faster deterministic algorithm was proposed by Wulff-Nilsen \cite{wulff2013faster}, of which the update time is $O(\lg ^2 n/\lg\lg n)$. Turning to the worst-case dynamic connectivity, a deterministic $O(\sqrt{n})$ update-time algorithm is Frederickson's $O(\sqrt{m})$ worst-case algorithm \cite{frederickson1985data} sped up via \emph{sparsification} technique proposed by Eppstein et al.~\cite{eppstein1997sparsification}. The result holds for online updating of minimum spanning trees. With roughly the same structure, but different and simpler techniques, Kejlberg-Rasmussen et al.~\cite{DBLP:journals/corr/Kejlberg-Rasmussen15} provided the so far best deterministic worst-case bound of $O(\sqrt{n(\lg\lg n)^2/\lg n})$ for dynamic connectivity. After the discovery of $O(\sqrt{n})$ update-time algorithm, people were wondering whether any poly-logarithmic worst-case update time algorithm is possible, even randomized. The open problem stands firmly for many years. A breakthrough should be attributed to Kapron et al.~\cite{kapron2013dynamic}. Their algorithm is Monte-Carlo, with poly-logarithmic worst-case update time. It has several improvements until now, as done in \cite{DBLP:journals/corr/GibbKKT15,DBLP:journals/corr/Wang15w}. For subgraph connectivity, the trivial update time of $\widetilde{O}(n)$ follows from Kapron et al.'s algorithm. The query time of it for subgraph connectivity can also be improved to $O(1)$, as the explicit maintenance of connected components can be done without blowing up the $\widetilde{O}(n)$ update time.

\subsection{Our Results}
The former $\widetilde{O}(m^{4/5})$ deterministic worst-case \emph{subgraph} connectivity structure adopted as a sub-routine the $O(\sqrt{n})$ deterministic worst-case algorithm for dynamic \emph{graph} connectivity. Now the randomized poly-logarithmic worst-case connectivity structures for dynamic \emph{graph} connectivity are discovered. We consider the question of whether it brings progress in \emph{subgraph} connectivity. The answer is affirmative. However, it does not come by replacing the $O(\sqrt{n})$ deterministic algorithm with the poly-logarithmic randomized one. Carefully tuning the former setting of the $\widetilde{O}(m^{4/5})$ algorithm is also in vain. Intuitively, the amortized $\widetilde{O}(m^{2/3})$ update time was achieved partially because it uses the connectivity structure of poly-logarithmic amortized update time. Now poly-logarithmic worst-case algorithms are discovered, it seems that the $\widetilde{O}(m^{2/3})$ worst-case update time is in sight. Nonetheless, we found that it is still hard to get the $\widetilde{O}(m^{2/3})$ update time. Until now we obtain the update time of $\widetilde{O}(m^{3/4})$. Although it only improves the previous deterministic $\widetilde{O}(m^{4/5})$ bound by $m^{0.05}$, note that the gap between the two deterministic results $\widetilde{O}(m^{2/3})$ and $\widetilde{O}(m^{4/5})$ is only $m^{0.134}$. The main contribution is a new organization of the auxiliary data structures.

The $\widetilde{O}(\sqrt{mn})$ result comes from dynamic DFS tree \cite{baswana2016dynamic,DBLP:journals/corr/ChenDWZ16}, which is a periodic rebuilding technique with fault tolerant DFS trees. Note that the $\widetilde{O}(m^{3/4})$ update time of our result is always no worse than $\widetilde{O}(\sqrt{mn})$ as $n=\Omega(m^{1/2})$.

\begin{table}[t]
	\centering
	\caption{Results on Dynamic Subgraph Connectivity}
	\label{table_results}
	\begin{tabular}{ccc}
		\hline
		Update time & Query time & Notes\\
		\hline\hline
		$\widetilde{O}(m^{4\omega/(3\omega+3)})$ & $\widetilde{O}(m^{1/3})$ & 
		\begin{tabular}{c}
			Amortized, \\deterministic, linear space \cite{chan2002dynamic}
		\end{tabular}\\
		\hline
		$\widetilde{O}(\sqrt{mn})$ & $O(1)$ &
		\begin{tabular}{c}
			Worst case, \\deterministic, space $\widetilde{O}(m)$ \cite{baswana2016dynamic,DBLP:journals/corr/ChenDWZ16}
		\end{tabular}\\
		\hline
		$\widetilde{O}(m^{2/3})$ & $\widetilde{O}(m^{1/3})$ & 
		\begin{tabular}{c}
			Amortized, \\ deterministic, space $\widetilde{O}(m^{4/3})$ \cite{chan2011dynamic}
		\end{tabular}\\
		\hline
		$\widetilde{O}(m^{2/3})$ & $\widetilde{O}(m^{1/3})$ & 
		\begin{tabular}{c}
			Amortized, \\ deterministic, linear space \cite{duan2010new}
		\end{tabular}\\
		\hline
		$\widetilde{O}(m^{4/5})$ & $\widetilde{O}(m^{1/5})$ & 
		\begin{tabular}{c}
			Worst case, \\deterministic, space $\widetilde{O}(m)$ \cite{duan2010new}
		\end{tabular}\\
		\hline
		$\widetilde{O}(m^{3/4})$ & $\widetilde{O}(m^{1/4})$ &
		\begin{tabular}{c}
			Worst case, \\randomized, linear space, this paper
		\end{tabular}
	\end{tabular}
\end{table}

Faster query time can be traded with slower update time for the bottom four results in Table \ref{table_results}. As to our result, $\widetilde{O}(m^{3/4+\epsilon})$ update time and $\widetilde{O}(m^{1/4-\epsilon})$ query time can be implemented. The details concerning these trade-offs are interpreted in Appendix \ref{sec_extend_epsilon}. Note that the trade-offs are in \emph{one direction}, i.e.~better query time with worse update time, but not vice-versa. Consequently, the former $\widetilde{O}(m^{4/5})$ algorithm never gives update time of $\widetilde{O}(m^{3/4})$. The trade-off phenomenon is definitely hard to break, as indicated by the OMv (\emph{Online Boolean Matrix-Vector Multiplication}) conjecture proposed by Henzinger et al.~\cite{henzinger2015unifying}. The OMv conjecture rules out \emph{polynomial} pre-processing time algorithms with the product of amortized update and query time being $o(m)$. Finally, our result is grouped as the following theorem.

\begin{thm}[Main Theorem]
	Given a graph $G=(V,E)$, there is a data structure for the dynamic subgraph connectivity, which has the worst-case vertex (edge) update time $\widetilde{O}(m^{3/4})$, query time $\widetilde{O}(m^{1/4})$, where $m$ is the number of edges in $G$, rather than in the subgraph of $G$ induced by $S$. The answer to each query is correct if the answer is \lq\lq yes\rq\rq, and is correct w.h.p.~if the answer is \lq\lq no.\rq\rq~ The pre-processing time is $\widetilde{O}(m^{5/4})$, and the space usage is linear.
	\label{thm_main}
\end{thm}

\section{Preliminaries}\label{sec_pre}
We review the two results on dynamic graph connectivity, which are adopted as sub-routines in our data structure. The first one is deterministic, and the second one is randomized. Given $G=(V,E)$ with $n$ vertices and $m$ edges, the properties of the two data structures are described in the following two theorems.

\begin{thm}[\cite{DBLP:journals/corr/Kejlberg-Rasmussen15}]\label{thm_deterministic}
	A spanning forest $F$ of $G$ can be maintained by a deterministic data structure of linear space, with $O(\sqrt{{m(\lg\lg n)^2}/{\lg n}})$ worst-case update time for an edge update in $G$, and constant query time to determine whether two vertices are connected in $G$.
\end{thm}

\begin{thm}[\cite{kapron2013dynamic}]
	There is a randomized data structure on dynamic graph connectivity, which supports the worst-cast time $O(\lg^4 n)$ per edge insertion, $O(\lg^5 n)$ per edge deletion, and $O(\lg n/\lg\lg n)$ per query. For any constant $c$ the answer to each query is correct if the answer is \lq\lq yes\rq\rq~and is correct with probability $\ge 1-{1}/{n^c}$ if the answer is \lq\lq no.\rq\rq~The pre-processing time of it is $O(m\lg^3 n+n\lg^4 n)$.
	\label{thm_randomized}
\end{thm} 

For dynamic \emph{subgraph} connectivity, we first consider the case when the updates are only vertex updates. The extension to edge updates is deferred to Appendix \ref{sec_edge-updates}. Hence temporarily $G$ is assumed to be static, as $E$ does not change if there are no edge updates. The vertex updates change $S$. Initially, $G$ is slightly modified to keep $m=\Omega(n)$ during its lifetime, i.e., for every $v\in V$, insert a new vertex $v'$ and a new edge $(v,v')$. The variant graph has $m=\Omega(n)$, which is required in the definition of the dynamic subgraph connectivity \cite{chan2002dynamic}. It facilitates the presentation of time and space complexity as functions of $m$, rather than of $n$ in the case of degenerate
graphs.

\section{The Data Structure}\label{sec_alg}

We give some high-level ideas. Main difficulties are the update of $S$ (recall that $S$ is the set of active vertices) incurred by the high-degree vertices, as their degrees are too high to explicitly delete their incident edges one by one. Nonetheless, if the low-degree vertices had been removed, the graph became smaller, and consequently former high-degree vertices were not high-degree anymore. Hence our aim is to remove the low-degree vertices. After that, some artificial edges are added to restore the loss of connectivity due to the removal of the low-degree vertices. Next a dynamic connectivity data structure is maintained on the modified graph, i.e., the graph with the low-degree vertices removed, and the artificial edges added. Besides, as $S$ evolves dynamically, we need to update the artificial edges accordingly. Hence the \emph{point} is how to maintain these artificial edges consistently and efficiently. We now move to the details. We partition $V$ according to their degrees in $G$. Use $\text{deg}_G(v)$ to denote the degree of $v$ in $G$.

\begin{itemize}
\item $C:$ Vertices with $\text{deg}_G(v)>{m}^{1/2}$
\item $B$: Vertices with ${m}^{1/4}<\text{deg}_G(v)\le m^{1/2}$
\item $A$: Vertices with $\text{deg}_G(v)\le m^{1/4}$
\end{itemize}

Denote $C\cap S$, $B\cap S$, and $A\cap S$ as $V_C$, $V_B$, and $V_A$ respectively. Consider the subgraph $G_A$ of $G$ induced by $V_A$. Define the \emph{degree} of a component as the sum of $\text{deg}_G(v)$'s for $v$'s in it. According to the degrees of the components, partition the components of $G_A$ into two types: \emph{high} component, with its degree $>m^{1/4}$; \emph{low} component, with its degree $\le m^{1/4}$. A spanning forest $F_A$ of $G_A$ is maintained by  the deterministic connectivity structure of Theorem \ref{thm_deterministic}.

\subsection{Path Graph}
A path graph inserts some artificial edges to reflect the \lq\lq are connected\rq\rq~relation of the vertices within $V_B$ via directly linking with a component of $G_A$. The idea was proposed in \cite{duan2010new}, whereas here we describe it more rigorously, giving more details. W.l.o.g.~assume $V=\{0,\dotsc,n-1\}$. Consider a spanning tree $T$ of $F_A$.

\begin{itemize}
	\item \emph{subpath tree}: For $v\in T$, identify the set of vertices in $V_B$ that are adjacent to $v$. Store the set of vertices in a balanced search tree, which has the worst-case $O(\lg n)$ update time for the well-known search-tree operations \cite{Cormen:2009:IAT:1614191}. Name the search tree as the subpath tree of $v$. Given the subpath tree of $v$, a sequence of artificial edges is added to link the vertices stored in the subpath tree of $v$. The sequence of artificial edges constitutes a subpath. 
	
	\item \emph{path tree}: Given $T\in F_A$, group all $v\in T$ with the \emph{non-empty} subpath tree as a balanced search tree, ordered by the Euler-tour order of $T$ (See Appendix \ref{sec_notation}.). Name it as the path tree of $T$. As each vertex stored in the path tree of $T$ has an associated subpath, these subgraphs are also concatenated one by one via the artificial edges, generating a path. To emphasize its difference from an ordinary path, it is referred to as the \emph{path graph} of $V_B$ w.r.t.~$T$. An example is shown in Fig.~\ref{fig_path-graph}.
\end{itemize}

\begin{figure}[t]
	\includegraphics[scale=0.75]{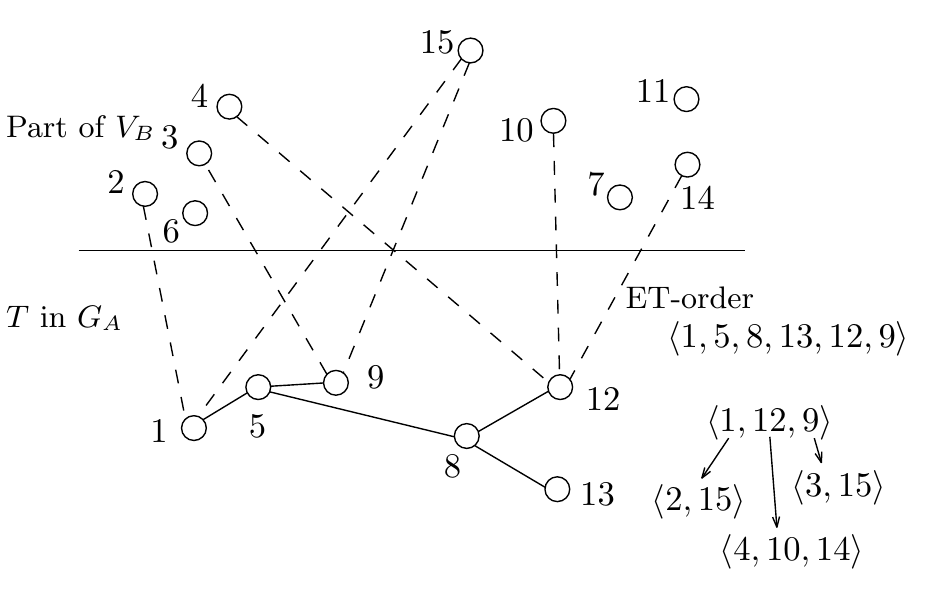}
	\centering
	\caption{The path graph of $V_B$ w.r.t.~a spanning tree $T$ in $F_A$. The \emph{dashed} edges represent edges between $V_B$ and $V_A$. The path tree is on sequence $\langle 1, 12, 9\rangle$, and three subpath trees are on sequences $\langle 2, 15\rangle$, $\langle 4, 10, 14\rangle$, and $\langle 3, 15\rangle$ respectively. The resulted path graph is a path $\langle 2,15,4,10,14,3,15\rangle$.}
	\label{fig_path-graph}
\end{figure}

\begin{lem}
	The path graphs can be updated in $\widetilde{O}(m^{1/2})$ time for a vertex update in $V_B$, and in $\widetilde{O}(1)$ time for a link or cut on $F_A$.
	\label{lem_path-graph}
\end{lem}

\begin{proof}
We categorize the analysis into two cases. 
\begin{itemize}
\item Reflect a vertex update in $V_B$: Suppose $v\in V_B$ is removed from $S$. The case of insertion is similar. $v$ has $\le m^{1/2}$ edges adjacent to $F_A$. Consider $(v,w)$ with $w\in T$. We locate $w$ in the path tree of $T$. Now the subpath associated with $w$ is known. Update the subpath of $w$ by removing $v$ from the subpath. If $v$ happens to be the first or the last vertex on the subpath, the path graph of $T$ is also updated. As the subpaths and the path graph are concerned with the nodes stored in the subpath trees and the path tree respectively, which are all balanced search trees, the removal of $(v,w)$ needs $\widetilde{O}(1)$ time. The removal of all such $(v,w)$'s requires $\widetilde{O}(m^{1/2})$ time.

\item Reflect a link or cut on $F_A$: We only discuss the edge cut on $F_A$. The edge link is similar. Assume the edge cut is $(v,w)\in T$, and the Euler tour of $T$ is $\langle L_1, (v,w), L_2, (w,v), L_3\rangle$ (The details on Euler tours can be found in Appendix \ref{sec_notation}.). After the cut of $(v,w)$, the Euler tours for the two resulted trees are $\langle L_1, L_3\rangle$ and $\langle L_2\rangle$. We can determine the first vertex $a$ and the last vertex $b$ of $\langle L_2\rangle$. With the order tree of $T$ presented in Appendix \ref{sec_notation}, the predecessor of $a$ and the successor of $b$ in the path tree of $T$ can be found in $O(\lg^2 n)$ time. With the predecessor and the successor, the path tree of $T$ is split. After the split, $O(1)$ edges in the path graph are removed to reflect the split of the path tree of $T$. As a conclusion, the path graph can be updated in $\widetilde{O}(1)$ time to reflect a link or cut on $F_A$.
\end{itemize}

\end{proof}

\subsection{Adjacency Structure}

Given $T\in F_A$ and $v\in C$, we want a data structure that provides the fast query of whether $T$ and $v$ are adjacent, i.e., whether an edge $(u,v)$ exists with $u\in T$. The adjacency structure presented in \cite{duan2010new} can answer such queries. Here we describe it more rigorously, giving more details. Assuming $v\in C$, the adjacency structure of $v$ contains the following search trees.

\begin{itemize}
\item \emph{sub-adjacency tree}: Given $T\in F_A$, identify the set of vertices in $T$ that are adjacent to $v$. Store the set of vertices as a balanced search tree, ordered by the Euler-tour order of $T$ (See Appendix \ref{sec_notation}.). Name the balanced search tree as the sub-adjacency tree of $v$ w.r.t.~$T$.

\item \emph{adjacency tree}: Identify $T\in F_A$ by the smallest vertex in $T$. Group all $T\in F_A$, w.r.t.~which $v$ has non-empty sub-adjacency trees, as a balanced search tree. Name the balanced search tree as the adjacency tree of $v$.
\end{itemize}

The sub-adjacency trees and the adjacency tree of $v$ constitute the adjacency structure of $v$ w.r.t.~$F_A$. The query aforementioned is answered by checking whether $T$ is in the adjacency tree of $v$. Note $v\in C$, rather than $\in V_C$. The adjacency structure of $v\in C$ w.r.t.~$F_A$ is maintained even if $v\notin S$.

\begin{lem}\label{lem_adjacency-structure}
	The adjacency structures of $C$ w.r.t.~$F_A$ can be renewed in $\widetilde{O}(m^{1/2})$ time for a link or cut on $F_A$. Given a query of whether $v\in C$ is adjacent to $T\in F_A$, it can be answered in $\widetilde{O}(1)$ time.
\end{lem}

\begin{proof}
We only discuss the edge cut on $F_A$. The edge link is similar. The adjacency structures of the vertices in $C$ are renewed one by one. Consider $v\in C$. Suppose the edge cut occurs on $T$, splitting $T$ into $T_1$ and $T_2$. We check whether $T$ is in the adjacency tree of $v$. If \lq\lq no\rq\rq, the update is done; if \lq\lq yes\rq\rq, remove $T$ from it, and update the sub-adjacent tree of $v$ w.r.t.~$T$ to reflect the edge cut on $T$. For $T_j$ ($j=1,2$), add $T_j$ into the adjacent tree of $v$ if it is adjacent to $v$ (determined by whether a sub-adjacent tree of $v$ exists w.r.t.~$T_j$). For every vertex in $C$, we need to check and update when necessary. Hence the total update time is $\widetilde{O}(m^{1/2})$, since $|C|$ is $O(m^{1/2})$.

\end{proof}

\subsection{The Whole Structure}\label{sec_whole}
Now we turn to the discussion of the whole structure of our result. First, $V_A$ is removed. After that some artificial vertices and edges are added to the subgraph of $G$ induced by $V_B\cup C$, resulting in a graph $H$. (Note that we include the vertices in $C\setminus S$, rather than just $V_C$, which is $C\cap S$.) The artificial vertices and edges are used to restore the loss of connectivity due to the removal of $V_A$. Recall that the components of $G_A$ are either low or high. We describe how the artificial edges or vertices are added as follows.

\begin{itemize}
	\item Added by the path graphs: For $T\in F_A$, construct the path graph of $V_B$ w.r.t.~$T$.
	\item Added by the high components: For a high component $P\in G_A$, add a \emph{meta-vertex}. For $v\in C$ adjacent to $P$, add an artificial edge between $v$ and the meta-vertex. Identify the first vertex of the path graph of $V_B$ w.r.t.~$T$, where $T$ is the spanning tree of $P$. Add an artificial edge between \emph{the} first vertex and the meta-vertex.
	\item Added by the low components: For a low component $Q\in G_A$, construct a complete graph within the vertices in $C$ that are adjacent to $Q$. Similarly as above, identify the first vertex of the path graph of $V_B$ w.r.t.~$T$, where $T$ is the spanning tree of $Q$. Add the artificial edges between \emph{the} first vertex and the vertices in $C$ that are adjacent to $Q$.
\end{itemize}

After these, $H$ can be defined as follows.

\begin{itemize}
	\item The vertex set $V(H)$ of $H$: $V_B\cup C\cup M$, where $M$ is the set of meta-vertices. Since the degree of a high component is $> m^{1/4}$, and the vertices in $V_B\cup C$ are of degree $> m^{1/4}$, $H$ has $O(m^{3/4})$ vertices.
	\item The edge set $E(H)$ of $H$: The original edges of $G$ within $V_B\cup C$, and the artificial edges.
\end{itemize} 

Figure \ref{fig_whole-structure} gives an example for the construction. $H$ is a multigraph. Use $D[u,v]>0$ of edge multiplicity to represent the edge $(u,v)\in E(H)$. The discussion on the maintenance of $D[u,v]$'s is deferred to the end of the subsection. Now we construct a graph $G^*$, based on $H$.

\begin{figure}[t]
\includegraphics[scale=0.9]{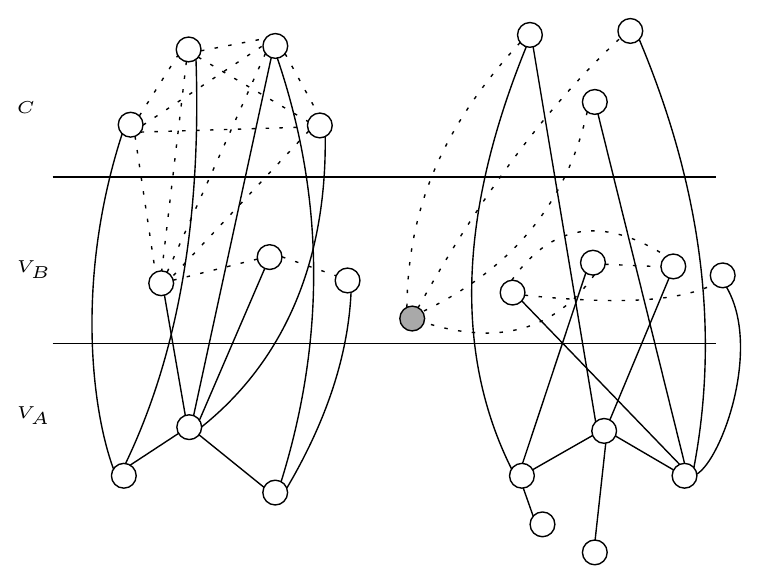}
\centering
\caption{An example of the whole structure. The irrelevant edges within $V_A$, $V_B$, and $C$ are omitted for clarity. The \emph{solid} edges are the edges in $G$, while the \emph{dotted} edges denote the artificial edges. The \emph{grey} vertex in the $V_B$ layer indicates a meta-vertex. The left component of $V_A$ is low; whereas the right one is high. We construct a complete graph within the vertices in $C$ w.r.t.~the low component.}
\label{fig_whole-structure}
\end{figure}

\begin{itemize}
	\item The vertex set $V(G^*)$ of $G^*$: $V_B\cup V_C\cup M$.
	\item The edge set $E(G^*)$ of $G^*$: The edges $(u,v)$'s with $D[u,v]>0$, where $u,v \in V(G^*), u\neq v$.
\end{itemize}

$G^*$ is a variant of the subgraph of $H$ induced by $V_B\cup V_C\cup M$. It excludes the vertices in $C\setminus S$, i.e., only the vertices in $V_C$ of $C$ are contained. Besides, the multiple edges are substituted by the single ones. $G^*$ is a simple graph. The randomized connectivity structure of Theorem \ref{thm_randomized} is maintained on $G^*$.

About the $D[u,v]$'s aforementioned, a balanced search tree is used to store them, with $D[u,v]$ indexed by $u+nv$ (assuming $u\le v$). Only $D[u,v]>0$ is stored in the search tree. We can check  whether $D[u,v]$ is in the search tree in $O(\lg n)$ time. Along the process of the updates, we might increment or decrement $D[u,v]$'s. When $D[u,v]$ decrements to 0, we remove it from the search tree. If both $u$ and $v$ are the vertices in $G^*$ and $u\neq v$, the edge $(u,v)$ is deleted from $G^*$. Similarly, when $D[u,v]$ increments to 1, we add it to the search tree. If both are the vertices in $G^*$ and $u\neq v$, the edge $(u,v)$ is inserted into $G^*$. $G^*$ captures the property of connectivity, which is stated in the following lemma.

\begin{lem}
	For any two vertices $u,v\in V_B\cup V_C$, they are connected in the subgraph of $G$ induced by $S$ if and only if they are connected in $G^*$.
	\label{lem_vbvc}
\end{lem}

\begin{proof}
	$G^*$ is a variant of the subgraph of $G$ induced by $S$. $G^*$ removes $V_A$ from \emph{the} subgraph. Connectivity within $V_B$ via $V_A$ is restored by the path graphs. Connectivity within $V_C$ via $V_A$ is restored \emph{either} by linking with the same meta-vertex, \emph{or} by the complete graph constructed. Lastly, for the connectivity between $V_C$ and $V_B$ via $V_A$, it is restored by the first vertex of the path graph linking with the meta-vertex, or with all the relevant vertices in $V_C$. Consider a path between $u$ and $v$ in the subgraph induced by $S$, the segments of the path consisting only of the vertices in $V_A$ can be eliminated, as the \lq\lq via $V_A$\rq\rq~connectivity is restored as discussed. The lemma follows.
	
\end{proof}

\subsection{Update and Query}

We discuss how the vertex updates are reflected efficiently in the whole structure constructed. The difficulty is to keep $D[u,v]$'s being consistent with $S$. As $E(G^*)$ is a subset of the $(u,v)$'s with $D[u,v]>0$, it might also need to be updated. We present the query algorithm and analyze it at the end of the subsection.

\begin{lem}\label{lem_update_query_time}
	The whole structure constructed has the worst-case vertex update time $\widetilde{O}(m^{3/4})$.
\end{lem}

\begin{proof}

We discuss the various cases of vertex updates, categorized according to whether $v\in A$, or $\in B$, or $\in C$.

\begin{itemize}
	
\item $v\in A$: Consider the case of inserting $v$ into $S$. $v$ is first inserted as a singleton component containing only $v$ in $G_A$. Next the edges incident on $v$ are restored in the following order: First, the edges between $v$ and $C$; second, the edges between $v$ and $V_B$; third, the edges between $v$ and $V_A$.

Restore the edges between $v$ and $C$: For every $u$ adjacent to $v$ where $u\in C$, construct a sub-adjacency tree (containing only $v$) of $u$, and insert $v$ into the adjacency tree of $u$. Next the complete graph within these $u$'s in $C$ is constructed. Because $\text{deg}_G(v)\le m^{1/4}$, i.e.~a low component, the update time is $\widetilde{O}(m^{1/2})$, dominated by constructing the complete graph.

Restore the edges between $v$ and $V_B$: Construct the subpath tree and the path tree of $v$. Add the path-graph edges associated with $v$ (\emph{Add} means incrementing the corresponding entry $D[u,v]$), and the edges \emph{between} the first vertex of the path graph \emph{and} the vertices in $C$ that are adjacent to $v$. The update time is $\widetilde{O}(m^{1/4})$.

Restore the edges between $v$ and $V_A$: $\widetilde{O}(\sqrt{m})$ deterministic data structure maintaining $F_A$ is updated in $\widetilde{O}(m^{3/4})$ time. As $\text{deg}_G(v)\le m^{1/4}$, the link or cut on $F_A$ happens $O(m^{1/4})$ times. Consequently, according to Lemma \ref{lem_path-graph}, the path graphs are updated in $\widetilde{O}(m^{1/4})$ time. According to Lemma \ref{lem_adjacency-structure}, the adjacency structures are updated in $\widetilde{O}(m^{3/4})$ time.

$O(m^{1/4})$ components of $G_A$ are affected. For every high component, using the adjacency structures, the edges between $C$ and the meta-vertex (corresponding to the high component) can be determined in $\widetilde{O}(m^{1/2})$ time according to Lemma \ref{lem_adjacency-structure}, since $|C|=O(m^{1/2})$; for every low component, as the degree of a low component is $\le m^{1/4}$, $\widetilde{O}(m^{1/2})$ time suffices to construct the complete graph within the vertices in $C$ that are adjacent to the low component, and $\widetilde{O}(m^{1/4})$ time suffices to construct the edges \emph{between} the first vertex of the path graph w.r.t.~the low component \emph{and} the vertices in $C$ that are adjacent to the low component. Hence no matter whether the component is low or high, the update time is $\widetilde{O}(m^{1/2})$. The time needed to update all these components is $\widetilde{O}(m^{3/4})$. Deleting of $v\in S$ from $S$ is a reverse process. In summary, a vertex update of $v\in A$ requires $\widetilde{O}(m^{3/4})$ time.

\item $v\in B$: Consider the case when $v\in S$ is removed. The case of insertion is the reverse. First destroy the edges between $v$ and $V_A$. According to Lemma \ref{lem_path-graph}, the path graphs can be updated in $\widetilde{O}(m^{1/2})$ time. Besides, $v$ might be the first vertex of some path graphs. We see how it is updated. $v$ can be adjacent to $\le m^{1/2}$ components of $G_A$, as $\text{deg}_G(v)\le m^{1/2}$. For a high component, as only one edge linking $v$ with the meta-vertex, the update is easy; for a low component, since only $\le m^{1/4}$ edges can be outward for a low component, $\widetilde{O}(m^{1/4})$ time suffices for updating the edges between $v$ and the vertices in $C$ that are adjacent to the low component. Hence the update time for $v$ being the first vertex of some path graphs is $\widetilde{O}(m^{3/4})$. Until now the artificial edges concerning $v$ are removed. Other edges concerning $v$ are the original edges in $G$. Hence we can remove these original edges one by one in $\widetilde{O}(m^{1/2})$ time as $\text{deg}_G(v)\le m^{1/2}$. In summary, the total update time of $v\in B$ is $\widetilde{O}(m^{3/4})$.

\item $v\in C$: As there are only $O(m^{3/4})$ vertices in $G^*$, the update time is $\widetilde{O}(m^{3/4})$. The relevant $D[u,v]$'s are left intact, and the adjacency structure of $v$ is not destroyed (if $v$ is removed from $S$). The total update time is $\widetilde{O}(m^{3/4})$.

\end{itemize}

\end{proof}

Now we describe the query algorithm: Given $u,v\in S$, the goal is to substitute them with the \emph{equivalent} vertices in $G^*$, where an \emph{equivalent} vertex of $u$ (or $v$) is a vertex in $G^*$ that is connected with $u$ (or $v$). As $V(G^*)=V_B\cup V_C\cup M$, if $u,v\in V_B\cup V_C$, the search for the equivalent vertices is done. Otherwise, if $u$ (or $v$) is in a high component, replace $u$ (or $v$) with the meta-vertex corresponding to the high component; if $u$ (or $v$) is in a low component, exhaustively search the outward edges of the low component for a vertex of $G^*$. When the equivalent vertex of $u$ (or $v$) cannot be found, it indicates that $u$ (or $v$) is in a low component of $G_A$, and the low component is not connected with any vertex in $V_B\cup V_C$. Intuitively $u$ (or $v$) is on an \lq\lq island\rq\rq~of $G_A$.

\begin{lem}\label{lem_correctness}
The time complexity of the query algorithm is $\widetilde{O}(m^{1/4})$. The answer to every query is correct if the answer is \lq\lq yes\rq\rq, and is correct w.h.p.~if the answer is \lq\lq no\rq\rq.
\end{lem}

\begin{proof}

Connectivity within $G^*$ is answered by the randomized connectivity structure on $G^*$; whereas for the other cases, $u$ and $v$ are connected if and only if they are in the same component of $G_A$, of which the queries can be answered by the deterministic connectivity structure on $G_A$. The time complexity is dominated by the exhaustive search if $u$ (or $v$) is in a low component, and thus is $\widetilde{O}(m^{1/4})$.

The correctness can be analyzed as follows. If $u,v\in V_B\cup V_C$, it follows from Lemma \ref{lem_vbvc}; otherwise, for any one not in, we only replace it with an equivalent vertex of $G^*$. If such an equivalent vertex cannot be found, the queried vertex is on an island aforementioned of $G_A$. Then $u$ and $v$ are connected if and only if they are on the same island. We analyze the error probability. A deterministic connectivity structure is adopted for $G_A$. $F_A$ is always a spanning forest of $G_A$. The queries are answered \emph{either} by the deterministic connectivity structure if at least one queried vertex is on an island aforementioned of $G_A$, \emph{or} by the randomized connectivity structure if both queried vertices are (replaced with) the vertices in $G^*$. The deterministic connectivity structure always gives the right answer; whereas the randomized one might answer erroneously. The randomized algorithm of \cite{kapron2013dynamic} maintains a private \emph{witness} of a spanning forest of $G^*$. The algorithm has the property that after every update, the witness is a spanning forest of $G^*$ with probability $\ge 1-1/n^c$. It is the property which ensures the answers are correct w.h.p.. Here, after every vertex update (which is transformed into a sequence of edge updates in $G^*$), the witness for $G^*$ is also a spanning forest of $G^*$ w.h.p.~after the vertex update. We can just focus on the correctness of the witness at the point after the last transformed edge update. Consequently, the error probability is negligible, i.e., $\le 1/n^c$ for any constant $c$.

\end{proof}

The analysis of space and pre-processing time is presented in Appendix \ref{app_space}. The main theorem, i.e.~Theorem \ref{thm_main}, follows from Lemmas \ref{lem_update_query_time}, \ref{lem_correctness}, \ref{lem_space_preprocess}, and \ref{lem_edge_update}.

\bibliographystyle{plain}
\bibliography{v2}

\appendix

\section{Trees as Euler Tours}\label{sec_notation}

An \emph{Euler tour} (\emph{cycle}) for a digraph is a cycle that traverses each edge exactly once, although it may visit a vertex more than once. For a strongly connected digraph, it has an Euler tour if and only if the \emph{in-degree} and the \emph{out-degree} are equal for each vertex. We adopt the Euler tours to represent the dynamic trees. The representation was first proposed by Miltersen et al.~\cite{miltersen1994complexity} and independently by Henzinger and King \cite{henzinger1999randomized}.

Given a tree $T$, following the idea of Tarjan's \cite{tarjan1997dynamic}, we replace each edge $\{v,w\}$ by two directed arcs $(v,w)$ and $(w,v)$, and add a self-loop $(v,v)$ for each node $v$. The resulted digraph always has at least one Euler tour. Choose any one, and break it at an arbitrary place. Store the resulted list of directed arcs as a balanced search tree, which is called the Euler-tour tree, a.k.a.~ET-tree.

The self-loop $(v,v)$ in the ET-tree actually represents $v\in T$. We can extract the order of these self-loops, and represent the order in another balanced search tree. Name it the \emph{order tree}, which generates an Euler-tour ordering of the nodes of $T$. We have the following observation.

\begin{obs}
	Given $u,v\in T$, the order of them w.r.t.~the Euler-tour order can be determined by the order tree in $O(\lg n)$ worst-case time.
\end{obs}

As shown in \cite{tarjan1997dynamic}, the ET-tree of $T$ has $O(\lg n)$ worst-case update time for a \emph{link} or \emph{cut} operation on $T$. We show briefly how to get the order trees for the two resulted trees after an edge cut. The edge link is similar. Suppose $cut(\{v,w\})$, and the Euler tour of $T$ is $\langle L_1, (v,w), L_2, (w,v), L_3\rangle$. The Euler tours of the two resulted trees should be $\langle L_1, L_3\rangle$ and $\langle L_2\rangle$. Augment the ET-tree to enable us to find the first and the last self-loop edges in $O(\lg n)$ time. (It is achieved by a boolean bit at each node of the ET-tree indicating whether there exists a self-loop in the subtree rooted at the search tree node.) We query the ET-tree of $\langle L_2\rangle$ to get the first node $a$ and the last node $b$. The order trees for the two resulted subtrees can be obtained by splitting the initial order tree first before $a$ and then after $b$, and joining the first part and the last part. The update time is easily observed as follows.

\begin{obs}
	ET-trees and order trees can be updated in $O(\lg n)$ worst-case time for a link or cut on dynamic trees.
\end{obs}

\section{Space Complexity and Pre-processing Time}\label{app_space}

The space usage consists of the deterministic connectivity structure on $G_A$, the search tree storing $D[u,v]>0$, the randomized connectivity structure on $G^*$, and the auxiliary structures, e.g., the adjacency structures. The space complexity and pre-processing time are summarized as the following lemma. From the proof we observe that the pre-processing time is dominated by $D[u,v]$ entries constructed for the complete-graph edges within $C$.

\begin{lem}\label{lem_space_preprocess}
The pre-processing time of the dynamic subgraph connectivity data structure is $\widetilde{O}(m^{5/4})$, and its space usage is linear.
\end{lem}

\begin{proof}
We analyze the space complexity and the pre-processing time one by one.
	
\begin{itemize}
	\item \textbf{Space complexity}
	
	The deterministic connectivity structure on $G_A$ occupies the space proportional to the size of $G_A$. As the edges in $G_A$ are all the original edges in $G$, it uses linear space.
	
	For the $D[u,v]$'s, the entries corresponding to the original edges in $G$ conform to the goal of linear-space usage. As to the artificial edges introduced by the path graphs, note that every original edge in $G$ between a vertex of $V_B$ and a vertex of $V_A$ is replaced by at most two edges within the path graphs. For the artificial edges connecting the first vertex of a path graph with some vertices of $C$, each one of them is brought because there is an original edge linking a vertex of $C$ with a low component in $G_A$. The edges between the vertices of $C$ and the meta-vertices can be accounted for in the same way. Until now every original edge in $G$ is attributed $O(1)$ times. Hence the space usage excluding the complete graphs constructed is linear. Considering the excluded edges constructed for the complete graphs, as only multiplicities are recorded, and $|C|=O(m^{1/2})$, $O(m)$ space suffices. In summary, $D[u,v]$'s size is linear.
	
	For the randomized connectivity structure on $G^*$, its space usage is $O(m^\prime+n^\prime\lg ^3 n^\prime)$, for a graph with $m^\prime$ edges, $n^\prime$ vertices. (The $\lg^3 n^\prime$ factor is due to $O(\lg n^\prime)$ \emph{cutset} data structures, which are represented by ET-trees maintaining $O(\lg ^2 n^\prime)$ values. See \cite{kapron2013dynamic} for the details.) The edge set of $G^*$ is a subset of $D[u,v]$'s, therefore its cardinality is $O(m)$. As to $n^\prime$, its value for $G^*$ is $O(m^{3/4})$. In summary, the space usage of the randomized connectivity structure is linear.
	
	For the other auxiliary data structures, e.g., the adjacency structures, every appearance of a vertex of $V_A$ in a sub-adjacency tree corresponds to an original edge in $G$ between $C$ and $V_A$. Thus it is linear. In summary, combining all of these, we can see that the total space complexity is linear.
	
	\item \textbf{Pre-processing time}
	
	As to the pre-processing time, the deterministic and randomized connectivity structures both need $\widetilde{O}(m)$ time. Similar to the space usage, excluding the complete graphs, every edge between $V_A$ and $V_B\cup C$ induces $O(1)$ artificial edges, and $\widetilde{O}(1)$ pre-processing time suffices. 
	
	For the complete graphs, since for every edge from a vertex of $C$ to a low component, $O(m^{1/4})$ edges are added in the complete graphs (the degree of a low component is $\le m^{1/4}$). Hence the pre-processing time for the complete graphs is $\widetilde{O}(m^{5/4})$. The total pre-processing time is dominated by the complete-graph constructions, and thus is $\widetilde{O}(m^{5/4})$.
\end{itemize}

\end{proof}

\section{Extension to Edge Updates}\label{sec_edge-updates}

The structure described in Sect.~\ref{sec_alg} can be extended to handle the edge updates. As mentioned in Sect.~\ref{sec_pre}, the data structure is constructed on a variant of $G$, which has $m=\Omega(n)$.

To insert an edge $(u,v)$, the simple trick of creating a new vertex $z$ adjacent to only $u$ and $v$, and then inserting $z$ would not work here \cite{chan2002dynamic,chan2011dynamic}. The point is that it might violate the type of $u$ (or $v$), i.e., which set ($A$, $B$ or $C$) it belongs to. On the same time, it also might violate the type of a component of $G_A$, i.e., whether the component is low or high. Instead, we adopt the following algorithm for the edge updates.

\begin{enumerate}[\bfseries Step 1]
	\item Delete $u$ and $v$ from $S$, if $u$ or $v$ is $\in S$.
	\item Update $E$.
	\item Change the type(s) of $u$ (and $v$). The type transition between $C$ and $B$ is the only case requiring the \emph{extra} processing. For the other cases, just mark the new type(s), e.g., if a former $A$ vertex $v$ becomes a $B$ type, just mark the type of $v$ as $B$. W.l.o.g., suppose the degree of a former $B$ vertex $u$ increases to $>m^{1/2}$, i.e.~a new $C$ vertex. Note that $\text{deg}_G(u)$ is just $\Theta(m^{1/2})$, as it is just above the threshold. We take two steps to deal with them: First, construct its adjacency structure, which can be accomplished in $\widetilde{O}(m^{1/2})$ time; second, insert $D[u,v]$ entries into its search tree. As $\text{deg}_G(u)$ is $\Theta(m^{1/2})$, the worst case is that it is adjacent to $\Theta(m^{1/2})$ distinct low components. $\widetilde{O}(m^{3/4})$ time suffices for calculating $D[u,v]$'s.
	\item Insert $u$ and $v$ back into $S$, if $u$ or $v$ is in $\in S$ before the edge update.
\end{enumerate}

The edge-update time is analyzed as follows.

\begin{lem}\label{lem_edge_update}
The data structure for the dynamic subgraph connectivity has the worst-case edge-update time $\widetilde{O}(m^{3/4})$, where $m$ is the number of edges in $G$, i.e.~$|E|$, at the time of the edge-update operation.
\end{lem}

\begin{proof}
	
One can easily see that the update time is $\widetilde{O}(m^{3/4})$, as Step 3 only needs $\widetilde{O}(m^{3/4})$ time, and the number of vertex updates is constant. Now we take a look at the value of $m$.

As $m$ is used as a parameter for the division of vertices and the partition of components, $m$ should not deviate from the true number of edges $m^\prime$ currently in graph $G$ by too much. So when detecting $|m^\prime-m|$ first $\ge\tfrac{1}{2}m^{1/2}$, we make a copy $G^\prime$ of the current graph $G$, and partition the vertices based on the value of $m^\prime$. During the next $\tfrac{1}{4}m^{1/2}$ updates, we pre-process $G^\prime$ in $\widetilde{O}((m^\prime)^{5/4})$ time, which is $\widetilde{O}(m^{5/4})$ as $m^\prime=\Theta(m)$. The pre-processing is done in the background, and is evenly distributed over the $\tfrac{1}{4}m^{1/2}$ updates. After the pre-processing, we start the catching-up process of performing $\tfrac{1}{2}m^{1/2}$ updates during the next followed $\tfrac{1}{4}m^{1/2}$ updates of $G$. By the end of $\tfrac{1}{2}m^{1/2}$ updates of $G$, we replace $G$ with $G^\prime$.

For the past $\tfrac{1}{2}m^{1/2}$ updates (supported by $G$), each only spends $\widetilde{O}(m^{3/4})$ worst-case update time on the pre-processing and the catching-up of $G^\prime$. So the worst-case update time is still $\widetilde{O}(m^{3/4})$. The crucial point is that for all the past updates, using $m$ is just fine, as the difference between $m$ and the true value is $\le m^{1/2}$, i.e., asymptotically these two values are the same.

Now we replace $G$ with $G^\prime$. Suppose the current number of edges is $m^{\prime\prime}$. We face a similar circumstance of $|m^{\prime\prime}-m^{\prime}|\le \tfrac{1}{2}(m^\prime)^{1/2}$ to the former $G$. Just begin the process above again, if at some point detecting $|m^{\prime\prime}-m^{\prime}|$ first $\ge\tfrac{1}{2}(m^\prime)^{1/2}$.

\end{proof}

\section{Extension to $\widetilde{O}(m^{3/4+\epsilon})$ Update and $\widetilde{O}(m^{1/4-\epsilon})$ Query Time}\label{sec_extend_epsilon}
We discuss how to adjust the parameters to achieve the query time of  $\widetilde{O}(m^{1/4-\epsilon})$. As before, partition the vertices according to the degrees.

\begin{itemize}
	\item $C$: Vertices with $\text{deg}_G(v)>{m}^{1-b}$
	\item $B$: Vertices with ${m}^{1-a}<\text{deg}_G(v)\le m^{1-b}$
	\item $A$: Vertices with $\text{deg}_G(v)\le m^{1-a}$
\end{itemize}

A component of $G_A$ is high if its degree is $>m^{1-c}$; otherwise, it is low. As a low component can be of degree as high as $m^{1-c}$, the query time should be $\widetilde{O}(m^{1-c})$. Consequently, we can anticipate the update time of $\widetilde{O}(m^c)$. To achieve the desired update time, the following inequalities should be satisfied.

\begin{itemize}
	\item Update in $A$:
	\begin{enumerate}[(1)]
		\item Maintaining a spanning forest of $G_A$: $\sqrt{m}\cdot m^{1-a}\le m^c$.
		\item Maintaining the complete graphs w.r.t.~the low components: $m^{1-a}\cdot m^{2(1-c)}\le m^c$.
		\item Maintaining the edges between $C$ and the meta-vertices: $m^{1-a}\cdot m^b\le m^c$.
	\end{enumerate}
	
	\item Update of a $B$ vertex $v$:
	\begin{enumerate}[(4)]
		\item $v$ can be adjacent to $\le m^{1-b}$ low components of $G_A$: $m^{1-b}\cdot m^{1-c}\le m^c$.
	\end{enumerate}
	
	\item Update in $C$:
	\begin{enumerate}[(5)]
		\item The number of vertices in $G^*$ should be bounded: $m^a\le m^c$.
	\end{enumerate}
	
\end{itemize}

From $(2)$ and $(5)$, we have $3\le a+3c\le 4c$. It gives the reason why the data structure achieves the update time no better than $\widetilde{O}(m^{3/4})$, as $c$ should be $\ge 3/4$. Nevertheless, it also tells us $c$ could be $3/4+\epsilon$. By setting $b$ to be $1/2$ and $a$ to be $3/4$, all inequalities are again satisfied. Hence we have a new data structure with $\widetilde{O}(m^{3/4+\epsilon})$ update time and $\widetilde{O}(m^{1/4-\epsilon})$ query time.

\end{document}